\documentclass{llncs}

\usepackage{makeidx}  

\usepackage{float}

\usepackage{amssymb}
\usepackage{graphicx}
\usepackage{float}
\usepackage{url}
\begin{document}

\title{Uniform Circle Formation by Transparent Fat Robots}
\titlerunning{Hamiltonian Mechanics}  
%
\author{Moumita Mondal \and Sruti Gan Chaudhuri}
%
%
%
\institute{Jadavpur University, Kolkata, India.
}

\maketitle              

\begin{abstract}
This paper addresses the problem of Uniform Circle Formation by $n>1$ transparent disc robots ({\it fat robots}). The robots execute repetitive cycles of the states {\it look-compute-move} in semi-synchronous manner where a set of robots execute the cycle simultaneously. They do not communicate by any explicit message passing. However, they can sense or observes the positions of other robots around themselves through sensors or camera. The robots are unable to recover the past actions and observations. They have no unique identity. The robots do not have any global coordinate system. They agree upon only y-axis (South and North direction). However, they do not have {\it chirality} or common  orientation of Y axis with respect to X axis.  Being transparent the robots do not cause any visual obstructions for other robots. But, they act as physical obstacles for other robots. This paper proposes a collision free movement strategy for the robots to form a uniform circle (in other words convex regular polygon) executing finite number of cycles. To the best of our knowledge this is the first reported results on uniform circle formation for fat robots under the considered model. 


\end{abstract}

\section{Introduction}

In this paper we will discuss about the algorithms for a group of mobile agents (also known as sensors or swarm robots) which execute a particular task together in co-operation. One of the current trends of research is to replace a big machine or robot by these mobile agents, and solve the given problem performing the job in coordination; such as guarding an area, moving an object, determining the shape of an object etc. Since a swarm of robots has minimal software and hardware complications with respect to installation and maintenance, they can be used in hostile environment where a big robot deployment is difficult. Moreover, the total cost of a group of tiny robots is less compared to a big robot. For executing a given job, often the primary task of the robots is to make some special geometric formation, obtained by their positions. This paper addresses one such geometric formation, uniform circle formation or in other words formation of a convex regular polygon by the mobile agents. 
 
\subsection{Framework}
The robots are represented as transparent unit discs. They act independently. Every robot executes a cycle consisting of three phases: {\it Look} - the robot takes a snapshot around itself and determines the other robots' positions w.r.t its own coordinate system; {\it Compute} - based on other robots' positions, the observer robot computes a destination point to move to; {\it Move} - the robot moves to the computed destination point. The robots execute this cycle under {\it semi-synchronous} scheduling where an arbitrary set of robots look, compute and move simultaneously. This scheduling assures that when a robot is moving no other robot is observing it. The robots do not stop before reaching its destination ({\it rigid} motion). The robots do not communicate by passing any explicit messages. The robots are indistinguishable, autonomous, oblivious (no recollection of computations and observations done in previous cycles). The robots do not have any global coordinate system and chirality or orientation. They only agree on the $Y$ axis. The direction of $X$ axis is not same for all the robots.
    The robots are transparent or see-through in order to ensure full visibility, but they act as physical obstructions for other robots. The robots move in such a way that after a finite time they are equidistantly apart on a circumference of a circle forming a convex regular polygon.

\subsection{Earlier Works}
Circle formation by mobile robots has been addressed by many researchers {\cite{Ref3}, \cite{Ref4}, \cite{Ref5}, \cite{Ref6}, \cite{Ref7}, \cite{Ref8}, \cite{Ref9}, \cite{Ref10} }.
Recently Flocchini et.al \cite{Ref11}, solved the uniform circle formation problem for point robots. It is shown that the Uniform Circle Formation problem is solvable for any initial configuration of n$\ne$4 robots without any additional assumption.
Mamino and Viglietta \cite{Ref12} solved the uniform circle formation for four point robots  thus completing the uniform circle formation problem for any initial configuration for point robots without any extra assumption. 

All of these algorithms assume that a robot is a point which neither creates any visual obstruction nor acts as an obstacle in the path of other robots. Czyzowicz et. al,\cite{Ref14} extended the traditional {\it weak} model of robots by replacing the point robots with unit disc robots (fat robots).
Dutta et. al \cite{Ref15} proposed a circle formation algorithm for fat robots assuming common origin and axes for the robots. Here the robots are assumed to be transparent in order to avoid visibility block. However, a robot acts as an physical barrier if it falls in the path of other robots. The visibility or sensing range of the robots is assumed to be limited.  
Datta et. al\cite{Ref16}, proposed another distributed algorithm for circle formation by a system of mobile asynchronous transparent fat robots with unlimited visibility where the robots can see upto a fixed region around themselves. 

\subsection{Contribution of this paper} Uniform circle formation has been solved for point robots \cite{Ref11,Ref12}. Circle formation for transparent fat robots has been addressed in \cite{Ref15,Ref16}. However, uniform circle formation for fat robots is not yet reported. In this paper, we propose an algorithm to form a convex regular polygon by transparent fat robots. The main concern in this algorithm is to avoid collisions among the robots. We show that if the robots are semi synchronous, execute rigid motion and agree on only one axis (e.g., $Y$ axis), then they can form a uniform circle without encountering collision.

\section{Algorithm Description}

A set of $n$ points on the 2D plane, representing the fat robots is given.  The robots are assumed to be transparent in order to ensure full visibility, but they act as physical obstruction for another robots. A robot is named by its center, i.e., by $Ri$ we mean a circular region with some finite radius around a point $Ri$ is a disc robot. 

The number of robots, $n$ and a length $a>3$ are given as the inputs of the algorithm. The length of the sides of the polygon will be at-least $a$. Since the robots are disc shaped with equal radius, this radius can  be used as a unit which is common to all robots. Hence, the robots can agree on the length $a$ of the polygon edge. The vertices of a convex regular polygon lie on a circle. The aim of the algorithm is to form this circle where the robots are placed equidistant apart on the circumference of this circle.

{\bf Free path:}
        A path of a robot is called free path, if from source to destination point (Ref to Fig. \ref{6})the rectangular area having length as the source to destination distance and width as two units, is not contained any part of another robot. A robot moves to its destination target point only by a free path, to ensure no collision.
        
        \begin{figure}[H]
        \centering
        \includegraphics[scale=0.5]{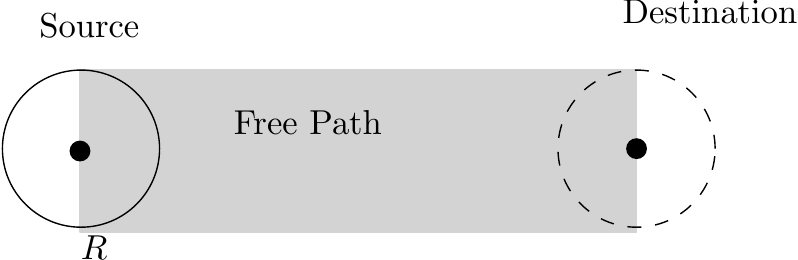}
        \caption{An example of free path of robot R }
        \label{6}
        \end{figure} 
        
{\bf Vacant Target Position:}
        A point is vacant if there exist no parts of another robot around a circular region of radius 1 around this point.

Following are the steps to be executed by each robot at their compute phase: 

\begin{itemize}
\item The robots compute the radius ($r$) of the circle to be formed using {\it ComputeRadius} routine. 
\item The robots compute $r_c$, the radius of current Smallest Enclosing Circle(SEC)\footnote{The circle with minimum radius such that all the robots are either inside the circle or on the circle.} of $n$ robots. 
\item If $r_c < r$, then a routine for expanding the SEC, {\it SECExpansion} is called.
\item Else a routine for forming uniform circle {\it FormUCircle} is called.  	
\end{itemize}

\subsection{ComputeRadius}
First the robot determine the radius of the circle to be formed. Let minimum radius of the circle required to accommodate all $n$ fat robots be $r$. The distance between two adjacent robots on the circle is given as $a$.  When there is no gap between two adjacent robots on the circle, the distance between the centres of two adjacent robots on the circle will be $2$ units. We assume that, $a$ is atleast $3$ units in length.
\begin{figure}[H]
\label{computeradi}
\centering
{\includegraphics[scale = 0.4,clip,keepaspectratio]{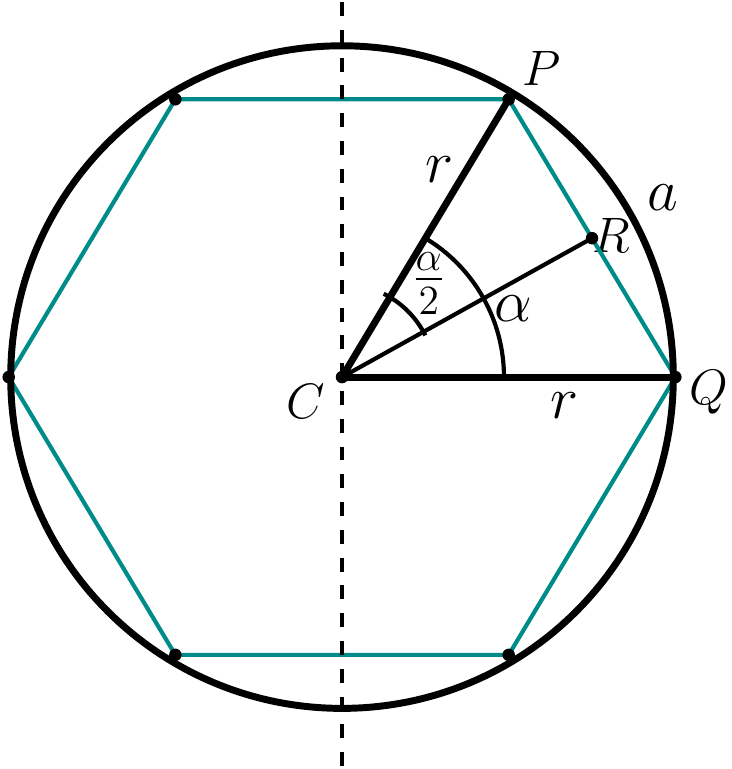}}
\caption{Minimum radius of SEC required to accommodate all the N robots .}
\end{figure}

Ref. to Fig. 1:
If $|PQ| = a$ and $PC=r$
then, 

\begin{equation}
\label{rad}
	r= \frac{a}{2sin(\frac{\alpha}{2})}
\end{equation}

where $\alpha=\frac{360}{n}$.

\subsection{SEC Expansion}
If this initial SEC cannot accommodate all the N fat robots (i.e. $r_c < r$), then the initial SEC is expanded by moving one or two robots on the SEC, such that radius of the new SEC, $r_n > r$. 
The steps of SEC expansion procedure is as below:

\begin{itemize}
	\item {\bf Step 1.} Under this procedure, first one or two leaders are elected for movement.  Let $L$ be the line parallel to the Y axis and passing through the center $c$ of the SEC. Two cases are possible and they are handled as follows:	
\begin{itemize}
	\item Case 1. The robot positions are not symmetric against the line $L$. The robot with maximum $Y$ value on the SEC and which has no mirror image against $L$, is elected as the leader. 
	
	\item Case 2. The robot positions are symmetric against $L$.	Two sub-cases are possible. 
	\begin{itemize}
	\item (a). There exists no robot on the north-most intersection point of $L$ and SEC. Note that, there exists two robots (say $R_{l1}$ and $R_{l2}$) on the SEC with maximum $Y$ value such that one robot is a mirror image of another robot against $L$. $R_{l1}$ and $R_{l2}$ both are elected as leaders. 
	\item (b). There exists a robot, $R_N$ on the north-most intersection point of $L$ and SEC. In this case $R_N$ is not selected as leader though it leads other robots by $Y$ value. There exists two robots (say $R_{l1}$ and $R_{l2}$) on the SEC with next maximum $Y$ value such that one robot is a mirror image of another robot against $L$. $R_{l1}$ and $R_{l2}$ both are elected as leaders. 
	\end{itemize}
	\end{itemize}

\item {\bf Step 2.} If the robots are in case 1, then draw a line $R_lc$ ($c$: center of the SEC). Let $R_lc$ intersect the SEC at $p$.  

If the robots are in case 2, then draw lines  $R_{l1}c$ and $R_{l2}c$. Let $R_{l1}c$ intersect the SEC at $p1$ and $R_{l2}c$ intersect the SEC at $p2$.

\item {\bf Step 3}. If the robots are in case 1 and there exists a robot, $R_p$  at $p$ (Fig. \ref{1}.), then $R_l$ moves $d$ distance, radially outward where $d = 2(r-r_c)$. Note that, the center of the changing SEC moves along the line joining $R_{p}$, $c$ and $R_{l}$.

\begin{figure}[H]
   \centering
   \includegraphics[scale=0.4]{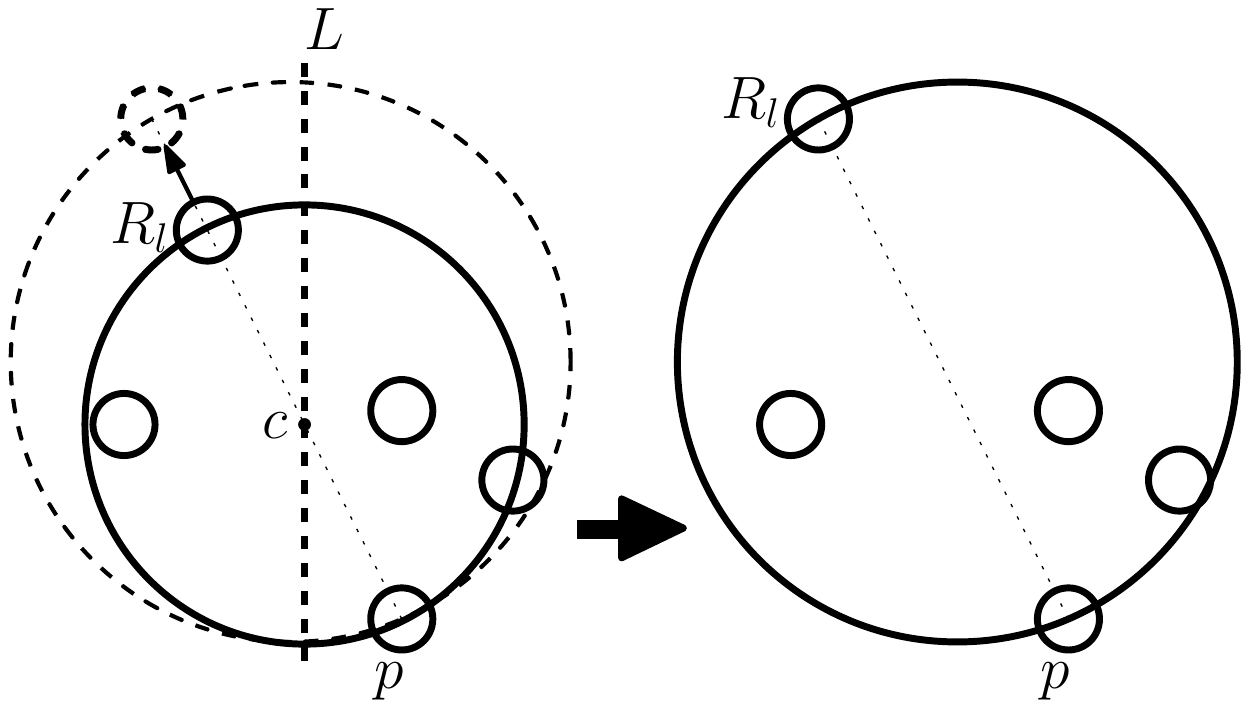}
   \caption{The robots are in case 1 and there exists robot at $p$ }
   \label{1}
   \end{figure}

Now consider the case when the robots are in case 2 and there exists a robot either at $R_{p1}$ or $R_{p2}$. Without loss of generality, let us assume that there exists a robot at $R_{p1}$. then $R_{l1}$ moves $d$ distance, radially outward where $d = 2(r-r_c)$. Note that, the center of the changing SEC moves along the line joining $R_{p1}$, $c$ and $R_{l1}$.

Next, consider the case when there exist robots at both  $R_{p1}$ or $R_{p2}$ (Fig. \ref{2}.). $R_{li} (i=1,2)$ moves $d$ distance, radially outward where $d = 2(r-r_c)$. 
Due to semi-synchronous scheduling both $R_{l1}$ and $R_{l2}$ may or may not execute the cycle simultaneously. Suppose they do not act simultaneously. $R_{p1}$ moves first. This situation is similar to case 1 of step 3. The center of the changing SEC moves along the line joining $R_{p1}$, $c$ and $R_{l1}$. On the other hand if both $R_{l1}$ and $R_{l2}$ act simultaneously, the center of the changing SEC moves along the bisector of the lines $R_{p1}c$ and $R_{p2}c$.  

 \begin{figure}[H]
  \centering
   \includegraphics[scale=0.4]{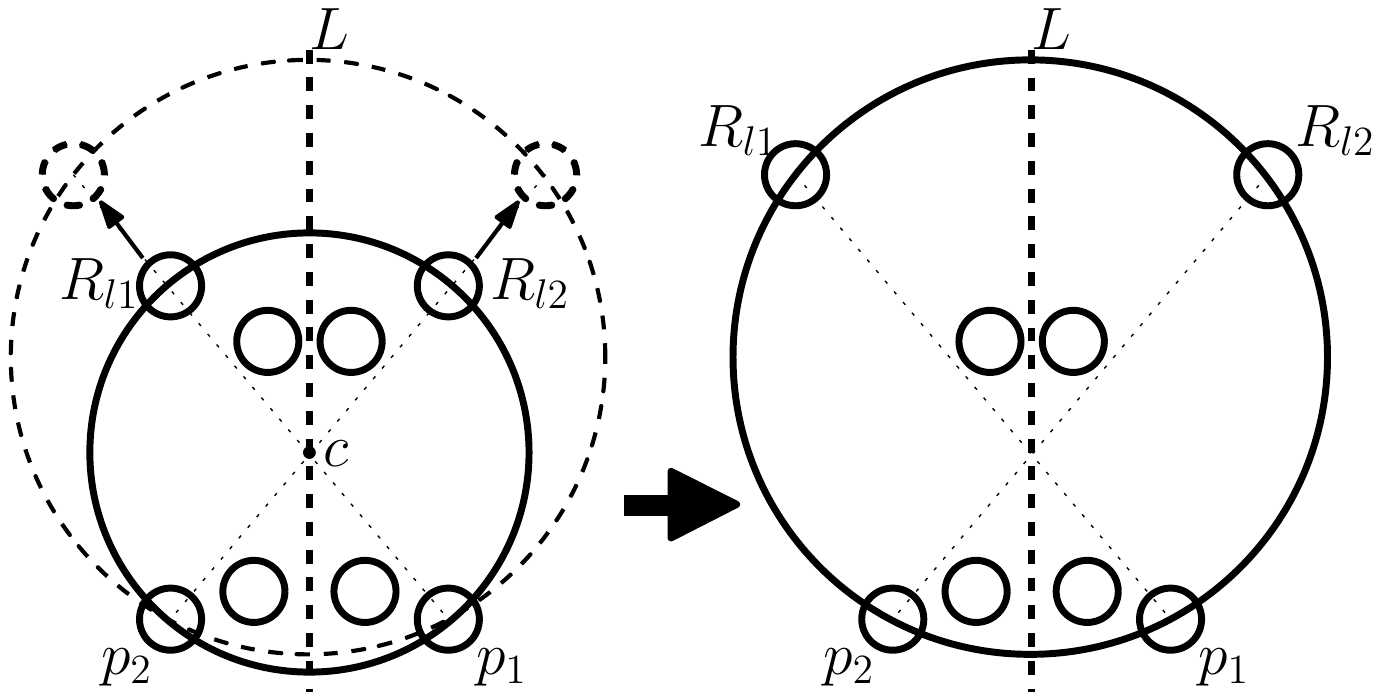}
   \caption{The robots are in case 2 and there exists robot at $p_{i} (i=1,2)$ }
   \label{2}
   \end{figure}

\begin{lemma}
When $R_l$ (case 1) or $R_{li} (i=1\ or\ 2)$ (case 2) is moving outwards, then $R_l$ (case 1) or $R_{li} (i=1\ or\ 2)$ (case 2) and $R_p$ (case 1) or $R_{pi} (i=1\ or\ 2)$ (case 2) always remain on the current SEC.
\end{lemma}
\begin{proof}
First consider case 1. 
Initially $R_l$ and $R_p$ are on the SEC and they are diagonally opposite to each other. Hence, $R_p$ is in maximum distance from $R_l$. No robot other than $R_l$  is moving. $R_l$ is moving following the straight line $R_pR_l$ and away from $R_p$. Thus the distance between $R_p$ and $R_l$ is increasing. $R_p$ continues to remain in maximum distant from $R_l$. According to the SEC property \cite{berg}, the maximum distant points of a set lies on the SEC of that point set. Hence, $R_l$ and $R_p$ remains on the current SEC (or on the changing SEC). 

Now consider case 2. 
If any of $R_{li1}$ or $R_{li2}$ is moving, the case is similar to case 1.
Otherwise, initially, $R_{pi} (i=1, 2)$ lie at diagonally opposite of $R_{li} (i=1, 2)$. Hence, $R_{pi} (i=1,2)$ is in maximum distance from $R_{li}(i=1,2)$. No robot other than $R_{li}(i=1,2)$  is moving. $R_{li}(i=1,2)$ is moving following the straight line $R_{pi}R_{li}$ and away from $R_{pi}$. Thus the distance between $R_{pi}$ and $R_{li}$ (for $(i=1,2)$) is increasing. $R_{pi}(i=1,2)$ continues to remain in maximum distant from $R_{li}(i=1,2)$. According to the SEC property the maximum distant points of a set lies on the SEC of that point set. Hence, $R_{li}$ and $R_{pi}$ remains on the current SEC (or in the changing SEC).    
\qed
\end{proof}

\begin{lemma}
\label{gerc0}
When $R_l$ (case 1) or $R_{li} (i=1,2)$ (case 2) reaches its destination, then the radius of the new SEC, $r_n \ge r_c$
\end{lemma}
\begin{proof}
Consider case 1. 
$R_l$ and $R_p$ in new position is the diameter of the new SEC. $R_l$moves in such a way that the length of the diameter of this SEC is  $2r_c$.
For case2, 
if any one robots among $R_{l1}$ or $R_{l2}$ moves, the situation is same as in case 1. Otherwise, the $R_{li} (i=1,2)$ and $R_{p_i}(i=1,2)$ lie on the new SEC and have maximum distance among all pairs of robots. Thus the radius of new SEC ($r_n$) is $>r_c$   
\qed
\end{proof}

\begin{lemma}
The movement of $R_l$(case 1) or $R_{li} (i=1,2)$ (case 2) is collision free.	
\end{lemma}
\begin{proof}
Since, $R_l$(case 1) or $R_{li} (i=1,2)$ (case 2) move diagonally outwards from the current SEC. No other robot is moving. Hence no robot comes in the path of these moving robots. Collision does not occur.
\qed
\end{proof}

\item If the robots are in case 1 and there exists no robot at $p$ (Fig. \ref{3}), then let $R_f$ be a robot on the SEC which is farthest from $R_l$. In case of tie the robot with maximum Y value is selected. Then, a line $R_fc$ is drawn. Compute a point $q$ on the ray ${R_fc}$, such that $|R_fq|= 2r$.
\begin{itemize}
\item If the path from $R_l$ to $q$ is a free path then $R_l$ moves to $q$.
\item Else, let $R'_l$ be the robot nearest to $q$, and has a free path to $q$. In case of tie, the robot with maximum Y value is considered as $R'_l$ and finally $R'_l$ moves to $q$. 
\end{itemize}

\begin{figure}[H]
   \centering
   \includegraphics[scale=0.4]{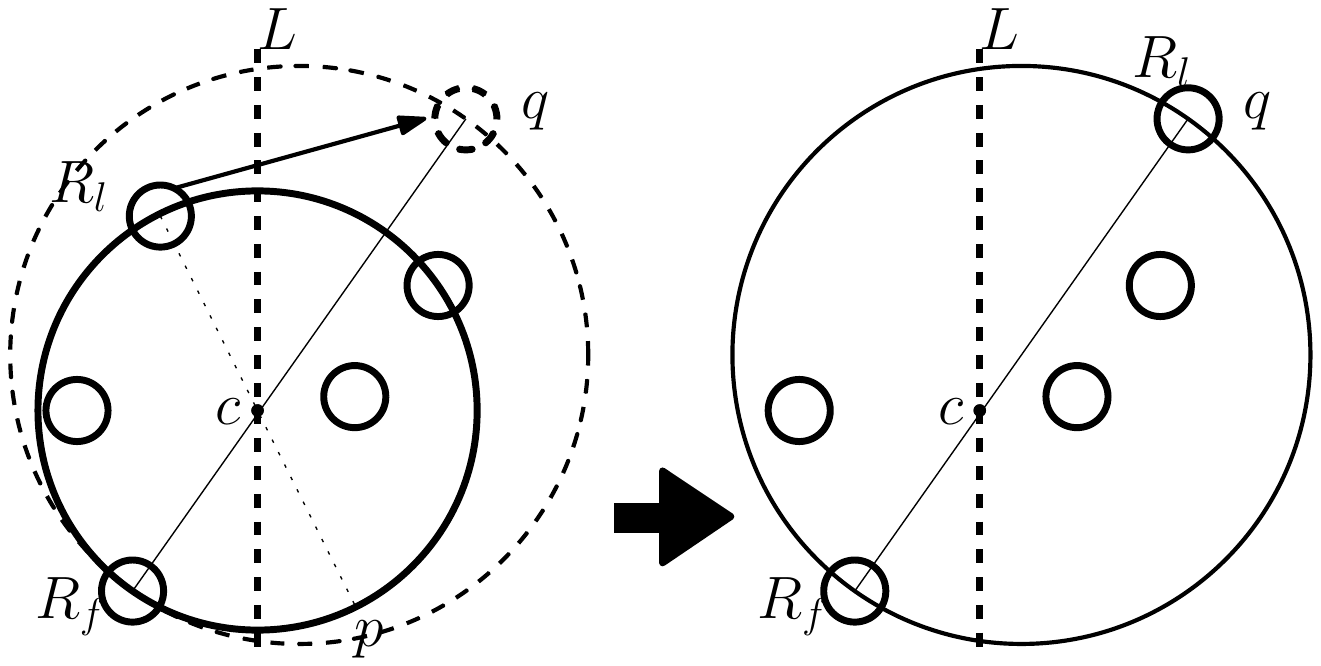}
   \caption{The robots are in case 1 and there exists no robot at $p$ }
    \label{3}
   \end{figure}

If the robots are in case 2 and there exists no robot at $p_{i} (i=1,2)$ (Fig. \ref{4}), then let $R_{fi}(i=1,2)$ be the robots on the SEC which is farthest from $R_{li}(i=1,2)$. If there exist two such farthest robots for each $R_{li}(i=1,2)$, then the one with the maximum Y value is considered. Then, the lines $R_{fi}c (i=1,2)$ are drawn. Compute the points $q_i(i=1,2)$ on the rays $R_{fi}c(i=1,2)$, such that $|R_fq_i|(i=1,2)= 2r$.

\begin{itemize}
\item If the path from $R_{li}(i=1,2)$ to $q_i(i=1,2)$ is a free path then $R_{li}(i=1,2)$ moves to $q_i(i=1,2)$.
\item Else, let $R'_{li}(i=1,2)$ be the robot nearest to $q_i(i=1,2)$, and has a free path to $q_i(i=1,2)$. If two such nearest robots exist for each  $q_i(i=1,2)$, then the one with the maximum Y value is considered as $R_{li}(i=1,2)$ and then finally, $R'_{li}(i=1,2)$ moves to $q_i(i=1,2)$. 
\end{itemize}

   \begin{figure}[H]
   \centering
   \includegraphics[scale=0.4]{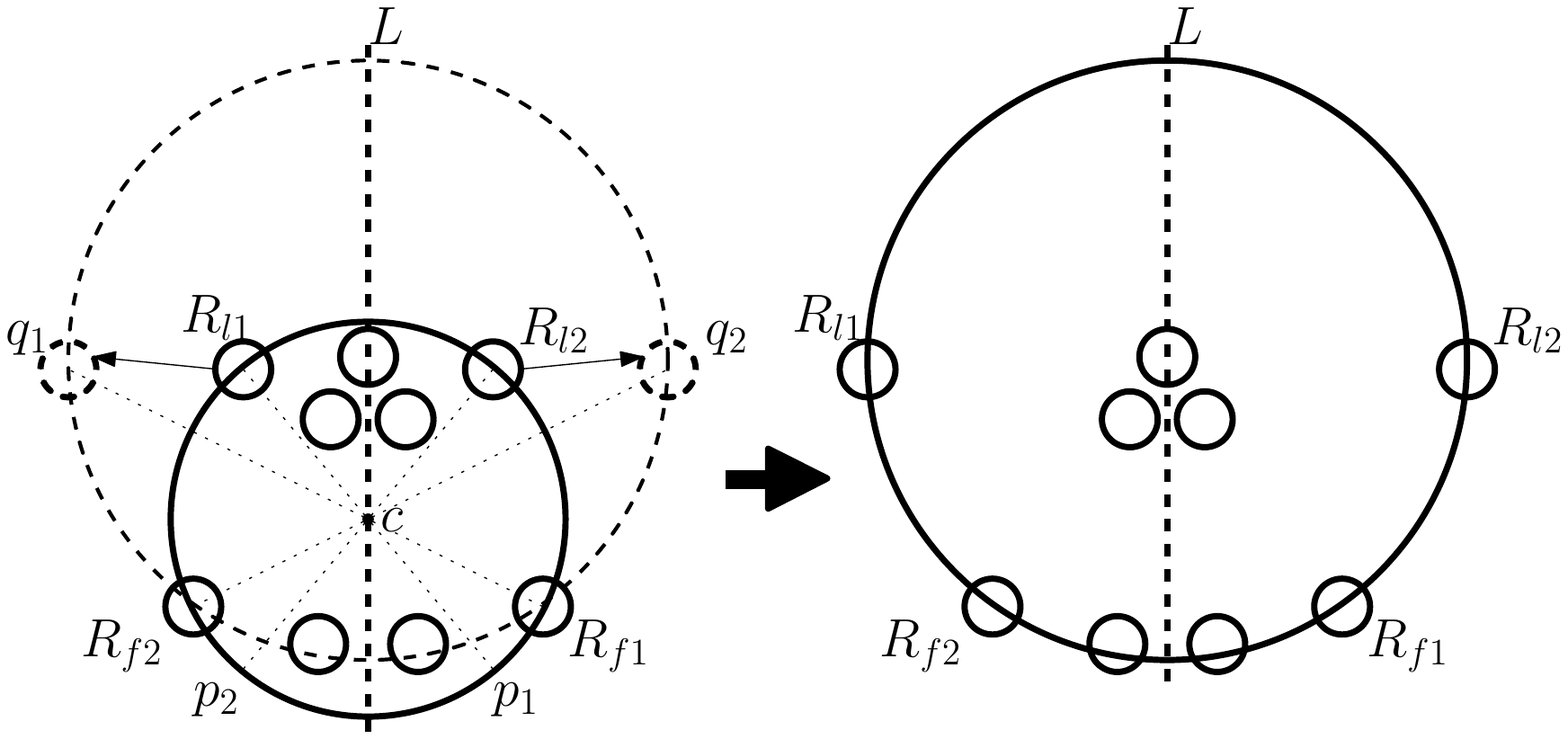}
   \caption{The robots are in case 2 and there exists no robot at $p_{i} (i=1,2)$ }
   \label{4}
   \end{figure}

 \begin{lemma}
 \label{onsec}
When $R_l$ or $R'_l$ (case 1) or $R_{li} (i=1,2)$ or $R'_{li} (i=1,2)$  (case 2) is (are) moving to $q$ (case 1) or $q_i(i=1,2)$ (case 2), then $R_l$ or $R'_l$ (case 1) or $R_{li} (i=1,2)$ or $R'_{li} (i=1,2)$ (case 2) and $R_f$ (case 1) or $R_{fi} (i=1,2)$ (case 2) always remain on the current SEC. 
\end{lemma} 
\begin{proof}
First consider case 1. Initially the distance between $R_l$ and $R_f$ is maximum. No robot other than $R_l$ or $R'_l$ moves. $|R_fq| > |R_fR_l|$ or $|R_fq| > |R_fR'_l|$. Hence, when $R_l$ or $R'_l$ reaches $q$ the distance between $R_f$ and $R_l$ or $R'_l$ remains maximum among other pairs of robots. According to the property of SEC the maximum distant points in a set lie on the SEC of that set. Hence $R_l$ or $R'_l$ and $R_q$ lie on the new SEC. 

Now consider case 2. 
Suppose the situation is when the path between $R_{li}$ and $q_{i}$ (i=1 or 2 respectively) are free paths. If any one of $R_{l1}$ or $R_{l2}$ moves, the case is similar to case 1. 
Otherwise, 
initially, $R_{fi} (i=1,2)$ lie at maximum distance from $R_{li} (i=1,2)$. No robot other than $R_{li}(i=1,2)$  is moving. 
$|R_{fi}{qi}| > |R_{fi}R_{li}| (i=1,2)$. Hence, when $R_{li}(i=1,2)$ reach $q_i(i=1,2)$ the distance between $R_{fi}(i=1,2)$ and $R_{qi}(i=1,2)$ remain maximum among other pairs of robots. According to the property of SEC the maximum distant points in a set lie on the SEC of that set. Hence $R_{fi}(i=1,2)$ and $R_{qi}(i=1,2)$ lie on the new SEC. 

Now consider the situation when the paths between $R_{li}$ and $q_{i}$ (i=1 or 2 respectively) are not free\footnote{Either one of the paths between $R_{l1}$ to $q1$ or $R_{l2}$ to $q2$ can be blocked by other robots or both the path are blocked but any one of $R'_{l1}$ and $R'_{l2}$ moves due to semi-synchronous scheduling}. 
Then $R'_{li}(i=1\ or\ 2)$ move(s) to $q_i(i=1\ or\ 2)$. Note that after the movement of $R'_{li}(i=1\ or\ 2)$ to $q_i(i=1\ or\ 2)$, the distance between $R'_{qi}(i=1\ or\ 2)$ and $R_{fi}(i=1\ or\ 2)$ is/are maximum among any other pair of robots. Hence, according to the property of SEC  $R'_{li}(i=1\ or\ 2)$ and $R_{pi}(i=1\ or\ 2)$ lie on the new SEC. 
\qed 
\end{proof}

\begin{lemma}
\label{gerc}
When $R_l$  or $R'_l$ reaches $q$ (case 1), 
$R_{li}(i=1,2)$ or $R'_{li}$ reaches $q_{i}(i=1,2)$ (case 2), the diameter of the new SEC $\ge r_c$  
\end{lemma}  
\begin{proof}
First consider case 1. The distance between $R_l$ or $R'_l$ at $q$ and $R_p$ is maximum among all pair distances. If $R_pq$ is the diameter of the new SEC then its radius is $= rc$. Otherwise, the radius of the new SEC is $ >r_c$.

Now consider case 2. Since $R_{pi} (i=1,2)$ and $q_i(i=1,2)$ are on the new SEC, with the similar argument as in case 1, it can be proved that the radius of new SEC $\ge r_c$
\qed
\end{proof}
 
\end{itemize}

\begin{lemma}
The movement of $R_l$ (case 1) or $R_{li}$ (case 2) to $q$ (case 1) or $q_i(i=1,2)$ (case 2) is collision free.	
\end{lemma}
\begin{proof} 
First consider case 1. 
According to the algorithm $R_l$ moves only when there is a free path to $q$. Otherwise, the robot $R'_l$ having free path to $q$ and nearest to $q$ moves to $q$. Thus there is no chance of collision as the robot moves along free path.

Case 2 can be proved using similar arguments. 
\qed
\end{proof}

 \begin{lemma}
If initially $r_c < r$, SECExpansion make $r_c >= r$ in finite time.
\end{lemma}
\begin{proof}
The leader robots move to enlarge the SEC. Since the robots are semi-synchronous the leaders do not change. The robots follow rigid motion, hence, the leaders successfully reach their destinations. following lemmas \ref{gerc0} and \ref{gerc} the cases in case 1 and 2, the radius of the new SEC is $ r_c\ge r$ .
\qed 
\end{proof}

\subsection{FormUCircle}
After achieving a big SEC, which can accommodate all the robots on it and having a minimum inter robot distance $a$, the robots execute FormUCircle subroutine.

\paragraph{FormUCircle}
\begin{itemize}
	\item The robots compute the target points on the SEC. These are computed as equidistant points, starting from the intersection point of $L$ and SEC as reference, as below:
	
        \paragraph{ComputeTargetPoint:} 
        Thus, for uniform circle to form, the distance between any 2 adjacent robot should be 2pi/N.\\
        Let the north-most ntersection point of $L$ and the SEC be $o$. 
Two cases can happen here:
\begin{itemize}
 \item \textbf{\normalsize{(i) SEC contains a robot at 'o':}}
 In this case, 'o', is the first target point. Next, on both sides of this first point, at a distance $\frac{2\pi}{N}$ apart on the SEC, the next two target points are determined. Similarly, the remaining (N-3) target points are also placed. This process continues until all the N distinct target points are determined. 
 

 \item \textbf{\normalsize{(ii) SEC does not contain a robot at 'o':}}
In this case, 'o' is not made a target point on the SEC. Instead, by taking the point 'o' on the SEC as the reference point, on both side of this reference point, at distance $\frac{\frac{2 \pi}{N}}{2}$  i.e.$\frac{\pi}{N}$ the first two target points are determined. Then from these two target points, the remaining (N-2) target points are decided such that on either side of every target point there is a target point at distance $\frac{2\pi}{N}$ apart. 


\end{itemize}

	\item If a target point $T'$ is partially occupied by a single robot $R'$, then $R'$ moves to $T'$.
	
	\item The target points are shorted with respect to the Y value. The target points are filled starting from the target points with maximum Y value. Let $T$ be a vacant target point having maximum Y value. 
	

	\item  The robot nearest to a target point moves to that target point. The robots on the SEC move to their corresponding target points by sliding on the SEC. Other  inside the SEC, move along straight lines towards the target points.
	\begin{itemize}
	\item The robot $R_T$ nearest to $T$, moves to $T$. If there are multiple robots nearest to $T$, the robot having maximum Y value among them moves to $T$. 
	\item If robot $R_T$ is the nearest robot for two vacant symmetric target points, then $R_T$ moves to the vacant target position nearest to it. If there is a tie between target points, $R_T$ moves to any of them.           
	\end{itemize}
		
    \item If there exists any robot  $R_o$ in the path of a robot, $R$, towards its target point $T$, then $R$, instead of moving in straight line, slides over $R_o$ to reach $T$.
    
    \begin{figure}[H]
   \centering
   \includegraphics[scale=0.6]{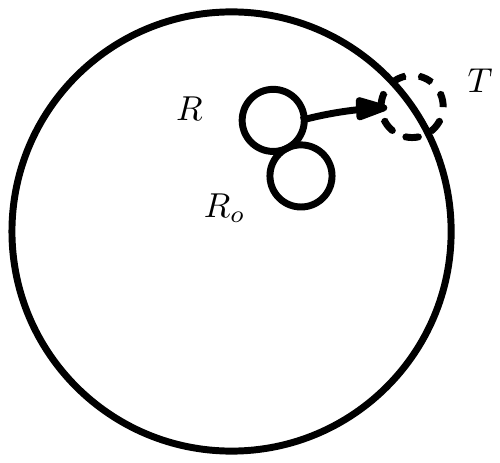}
   \caption{When there exists any robot  $R_o$ in the path of a robot, $R$}
   \label{8}
   \end{figure}



Note: If a vacant target point, $T'$ has its nearest robot, $R'$ on the SEC boundary such that : (1) $R'$ has the two adjacent target points $T1$ and $T2$ on the SEC already occupied by robots $R1$ and $R2$ respectively and (2) the destination target point, $T'$ has this intermediate target point robot, $R2$ as an obstruction in the path between $R'$ and $T'$; then robot $R2$ moves radially along the SEC to destination target point $T'$,  thus making target point $T2$ vacant. Then $R'$ also moves radially along the SEC to now vacant target point $T2$.

    \begin{figure}[H]
   \centering
   \includegraphics[scale=0.6]{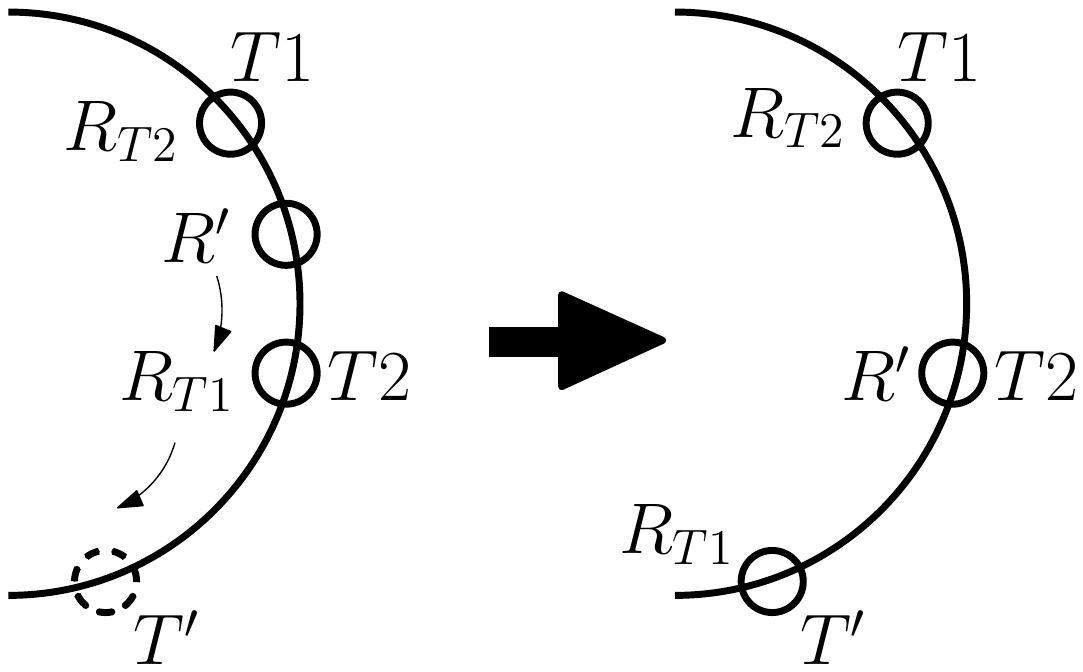}
   \caption{When a vacant target point, $T'$ has its nearest robot, $R'$ on the SEC boundary such that : (1) $R'$ has the two adjacent target points $T1$ and $T2$ on the SEC already occupied by robots $R1$ and $R2$ respectively and (2) the destination target point, $T'$ has this intermediate target point robot, $R2$ as an obstruction in the path between $R'$ and $T'$.}
   \label{8}
   \end{figure}
\end{itemize}

Since the robot nearest to the target point moves, following lemma holds.

\begin{lemma}
When a robot $R_T$ is moving to $T$, no other robot comes in its path, i.e., the movement of $R_T$ is collision free.	
\end{lemma}
\begin{proof}
In this algorithm, the robot nearest to the vacant north-most target point, moves to this target point. 
For existence of any obstacle robot following two situations may arise:
\begin{itemize}
\item The obstacle robot is nearer to the target point, which is not possible.
\item The moving robots can be obstructed by both sides when three robots are at same distance from the target and the middle robot touches other two robots from both the sides. 
In this situation the north-most robot is selected from the movement. This robot will have a open side and it will slide over the other robot and moves to its destination.
\end{itemize}
Thus the robots reach their destinations without collision.
\qed
\end{proof}

\begin{lemma}
There will be no deadlock in formation of regular polygon.	
\end{lemma}
\begin{proof} 
The vacant target points, starting from the north-most side, get filled by the robots. If no target point is vacant; it means, $n$ target points are partially or fully occupied by $n$ robots; Then all the robots will move to those target points occupied by them partially. Otherwise, since the number of target points is equal to the number of robots and no robot can partially block two target points, as being the side of the polygon is $>3$ units; there exists at least one vacant target point to be filled by its nearest robot.
Also, since, an ordering of the robots' movement is maintained in this algorithm, there is no deadlock and progress is guaranteed. 
\qed
\end{proof}

Through our algorithm, all the robots get their own target points and the collision free paths to reach them; thus forms an uniform circle.

\section{Conclusion}
Finally we can conclude the results in the following theorem.
\begin{theorem}
$n$ robots under semi-synchronous scheduling having rigid motion and agreement in one axis can form a uniform circle robots in finite time without collision.	
\end{theorem}

The immediate extension of this problem is to make the model weaker by considering asynchronous robots or non rigid movement or removing the agreement on axis. 

%
%


\end{document}